\documentclass[article,a4paper,11pt]{scrartcl}
\usepackage[english]{babel}
\usepackage[utf8x]{inputenc}
\usepackage{palatino}
\usepackage{multicol}
\usepackage{amsmath}
\usepackage{amssymb}
\usepackage{amsthm}
\usepackage{graphicx}
\usepackage{xcolor}
\usepackage{authblk}
\usepackage{enumitem}
\usepackage{marvosym}
\usepackage{cite}
\usepackage{mathtools}
\usepackage{vmargin}
\setmarginsrb{2cm}{2cm}{2cm}{2cm}{0mm}{0mm}{0mm}{10mm}
\usepackage{gensymb}
\usepackage[colorlinks=true, allcolors=blue, breaklinks=true]{hyperref}
\theoremstyle{definition}

\newtheorem{theorem}{Theorem}

\newtheorem{corollary}{Corollary}
\newtheorem{lemma}{Lemma}
\usepackage{url}
\newcounter{example}[section]

\title{Perturbed potential temperature distribution\\ in atmospheric boundary layers}
\author{\normalsize Natanael Karjanto\thanks{\Letter: \texttt{natanael@skku.edu} \href{https://orcid.org/0000-0002-6859-447X}{\includegraphics[scale=0.08]{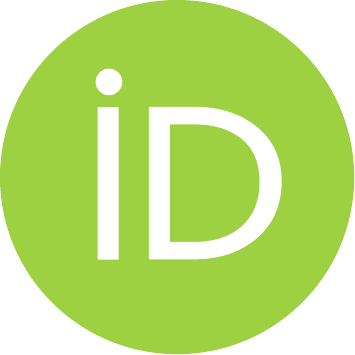}}}}
\affil{Department of Mathematics, University College, Natural Science Campus\\ Sungkyunkwan University, Suwon~16419, Republic of Korea}

\date{\vspace*{-0.5cm} \scriptsize Updated \today}

\begin{document}
\maketitle
\begin{abstract}
\begin{center}
{\bfseries Abstract}
\end{center}
\noindent
This article discusses the modeling of perturbed potential temperature in an atmospheric boundary layer. We adopt a convection-diffusion model with specified initial and boundary conditions that resulted from simplifying the linearized equation of the standard continuity equation for potential temperature field in the state of weak turbulent fluxes. By implementing the method of separation of variables to the non-steady-state perturbed potential temperature, we obtain a regular Sturm-Liouville problem for the spatial-dependent, vertical distribution component of the perturbed potential temperature. By transforming the canonical problem into the Liouville normal form, we provide asymptotic solutions for the corresponding second-order boundary value problem using the WKB theory. Furthermore, by solving the problem numerically, we observe a remarkable qualitative agreement between the asymptotic solutions and numerical simulations. As other convection-diffusion models typically perform, the perturbed potential temperature diminishes and approaches the steady-state condition over time.
\end{abstract}

\section{Introduction}
Although both the atmosphere and ocean play an important role in the variation of the climate system, it is the former that affects and interacts directly with human beings. Indeed, we influence and our lives are influenced by the lowest part of the atmosphere, known as the \emph{atmospheric} or \emph{planetary boundary layer} (ABL/PBL), where its characteristics are also strongly influenced by its contact with the Earth's surface. Hence, our understanding of how it works will help not only in short-term atmospheric conditions and local weather forecasts but also in long-term changes of climate predictions. 

Mathematical modeling of ABL involves fundamental aspects of geophysical fluid dynamics. Equations of motion cover the mass continuity and momentum equations for the atmospheric flow, which usually comprises velocity and potential temperature. In their influence paper on ABL, Owinoh et al. (2005) presented theoretical, computational, and experimental studies on the effects of changing surface heat flux on the flow and demonstrated the complexities in this transition process~\cite{owinoh2005effects}. In particular, they proposed a simplified model for perturbed layers by implementing scaling analysis to the standard Navier-Stokes equations for conservation of mass and momentum of Newtonian fluids. It follows that the perturbed potential temperature can be governed by a convection-diffusion model with appropriate initial and boundary conditions. 

The simplified two-dimensional model reads
\begin{equation*}
\frac{\partial \theta}{\partial t} - \frac{\partial }{\partial z} \left(\kappa u_\ast z \frac{\partial \theta}{\partial z} \right) = 0, \qquad z > 0, \qquad t > 0,
\end{equation*}
subject to boundary conditions far above the surface and at the surface layer, respectively
\begin{align*}
\theta(z,t) &= 0, \qquad \text{as} \quad z \to \infty \\
F_\theta(z_0,t) &= \left\{
\begin{array}{rc}
0, 	 & \qquad t < 0, \\
F_s, & \qquad t > 0.
\end{array}
\right. 
\end{align*}
Here, $\theta(z,t)$ denotes the potential temperature perturbation, $z$ is the vertical coordinate, $t$ is time, $z_0 > 0$ is the surface roughness length, $\kappa > 0$ is the von K\'{a}rm\'{a}n constant, a dimensionless constant involved in the logarithmic law describing the distribution of the longitudinal velocity in the wall-normal direction of a turbulent fluid flow near a boundary with a no-slip condition ($\kappa \approx 0.40 \pm 0.02$), $u_\ast$ is the friction velocity that depends on the shear stress $\tau_w$ at the boundary of the low, i.e., $u_\ast = \sqrt{\tau_w/\rho}$, with $\rho$ being the fluid density,  
\begin{equation*}
F_\theta(z,t) = - \kappa u_\ast z \frac{\partial \theta}{\partial z},
\end{equation*}
and $F_s$ is the constant perturbation surface heat flux initiated at $t = 0$.

By implementing a similarity transformation, i.e., $\zeta = z/(\kappa u_\ast t)$, and converting $\theta$ to a dimensionless quantity, i.e.,
\begin{equation*}
\theta^\ast(\zeta) = \frac{\kappa u_\ast}{F_s} \theta(z,t)
\end{equation*}
Owinoh et al.~\cite{owinoh2005effects} reduced the partial differential equation (PDE) to an ordinary differential equation (ODE) with appropriate boundary conditions (notation typos have been corrected): 
\begin{align*}
\zeta \frac{d^2\theta^\ast}{d\zeta^2} + (\zeta + 1) \frac{d\theta^\ast}{d\zeta} = 0 & \\
\zeta \frac{d \theta^\ast}{d \zeta} = -1 \qquad \text{as} \quad \zeta \to \frac{Z_0}{\kappa u_\ast t} \quad \text{at the surface} &\qquad \qquad \text{and} \qquad \qquad 
\theta^\ast \to 0 \qquad \text{as} \quad \zeta \to \infty.
\end{align*}
The solution of the corresponding boundary value problem above can be expressed in terms of the exponential integral function $E_1$:
\begin{equation*}
\theta(\zeta) = \frac{F_s}{\kappa u_\ast} E_1(\zeta), \qquad \text{where} \qquad E_1(\zeta) = \int_{\zeta}^{\infty} \frac{e^{-\zeta'}}{\zeta'} \, d\zeta'.
\end{equation*}
At the top of the thermal layer, the perturbed potential temperature profile can be described by the ratio between an exponential function and its argument:
\begin{equation*}
\theta(\zeta) = \frac{F_s}{\kappa u_\ast} \frac{e^{-\zeta}}{\zeta} \qquad \text{for} \quad \zeta \gg 1 \qquad \text{or} \qquad
\theta(z,t) = \frac{F_s t}{z} e^{-\frac{z}{\kappa u_\ast t}} \qquad \text{for} \quad z \gg t.
\end{equation*}
At the Earth's surface, $E_1(\zeta) \approx -\left(\gamma + \ln \zeta \right)$ as $\zeta \to 0$, and thus the surface perturbed potential temperature $\theta_s$ is given by
\begin{equation*}
\theta_s(\zeta) = \frac{-F_s}{\kappa u_\ast} \left(\gamma + \ln \zeta \right) \qquad \text{or} \qquad
\theta_s(z_0,t) = \frac{F_s}{\kappa u_\ast} \left[\ln \left(\frac{\kappa u_\ast t}{z_0} \right) - \gamma \right] \qquad \text{or} \qquad
\end{equation*}
Near the ground, the perturbed potential temperature profile is also expressed as the usual logarithmic profile:
\begin{equation*}
\theta(z,t) = \theta_s(z_0,t) - \frac{F_s}{\kappa u_\ast} \ln \left(\frac{z}{z_0}\right) = \frac{F_s}{\kappa u_\ast} \left[\ln \left(\frac{\kappa u_\ast t}{z}\right) - \gamma \right],
\end{equation*}
where $\gamma \approx 0.5772156649$ is the Euler-Mascheroni constant. It is defined as the limiting difference between the harmonic series and the natural logarithm:
\begin{equation*}
\gamma = \lim\limits_{n \to \infty} \left(-\ln n + \sum_{k = 1}^{n} \frac{1}{k} \right) = \int_{1}^{\infty} \left(-\frac{1}{x} + \frac{1}{\lfloor x \rfloor }\right) \, dx,
\end{equation*}
where $\lfloor x \rfloor$ denotes the floor function.

This article is organized as follows. After this introduction, Section~\ref{model} proposes a mathematical model that governs the perturbed potential temperature field in the atmospheric boundary layer. It follows that the spatial-dependent variable reduces to a regular Sturm-Liouville eigenvalue boundary value problem and we also derive its transformation from the canonical form to the Liouville normal form. Section~\ref{asymptotic} discusses asymptotic solutions of the eigenvalue problem using the WKB theory. We also discuss the case for a single turning point and several possible possibilities in estimating eigenvalues. Section~\ref{numerical} covers some findings of numerical simulations for the Sturm-Liouville problem and provides qualitative comparisons with the WKB asymptotic solutions. Finally, Section~\ref{conclusion} concludes our discussion.

\section{Mathematical model}	\label{model}

In this article, we focus on the following initial-boundary value problem (IBVP) for the perturbed potential temperature in an atmospheric boundary layer:
\begin{align}
\frac{\partial \theta}{\partial t} = \frac{\partial}{\partial z} \left(u(z) \frac{\partial \theta}{\partial z} \right),& 
\qquad \qquad t > 0, \quad 0 \leq z_0 < z < z_1, 		\label{pde-theta} \\	
\theta(z,0) = \psi(z),&  \qquad \qquad z_0 < z < z_1, 	\label{ic-theta} \\
a_0 \theta(z_0,t) - a_1 \frac{\partial \theta}{\partial z}(z_0,t) = c_1,& \qquad \qquad t > 0, \qquad a_0^2 + a_1^2 > 0, 	\label{bc1-theta}	\\
b_0 \theta(z_1,t) + b_1 \frac{\partial \theta}{\partial z}(z_1,t) = c_2,& \qquad \qquad t > 0, \qquad b_0^2 + b_1^2 > 0.	\label{bc2-theta}
\end{align}

We observe that the governing equation for $\theta(z,t)$ obeys the convection-diffusion model. While Owinoh et al.~\cite{owinoh2005effects} consider a linear wind speed profile, we adopt a generalized version of it. Hence, $u(z)$ describes the vertical distribution of horizontal mean wind speeds within the lowest portion of the atmospheric boundary layer. We assume that $u(z) \neq 0$ for all $z \geq 0$. This model is similar to the heat conduction problem in a rod whose material properties, e.g., the thermal conductivity, vary with the position~\cite{haberman2013applied,agarwal2009ordinary}. For the spatial dependent $u$, i.e., $u = u(z)$, the PDE is still linear. But if $u$ is temperature-dependent, i.e., $u = u(\theta)$, we have a nonlinear PDE. 

For the special case where the horizontal mean wind speed is constant, the PDE above reduces to the well-known classical heat equation, where the analog of this quantity in the heat problem model is usually called the \emph{thermal conductivity}. Furthermore, since thermal diffusivity is defined as the thermal conductivity divided by the product of the specific heat and mass density, we may also regard the case of constant wind speed being analog as the case of constant \emph{thermal diffusivity} in the classical heat equation problem, as long as the specific heat and mass density are also constant. In the context of diffusion of a chemical pollutant, the thermal conductivity is often called \emph{chemical diffusivity}.

While Owinoh et al.~\cite{owinoh2005effects} stated that their convection-diffusion model is valid in the inner part of the surface layer where $z/h \lesssim 1/10$, where $h$ is the boundary layer depth, their self-similar solution seems to extend not only to the outer part of the boundary layer, where $z/h \gtrsim 1/10$, but also to sufficiently large vertical position, i.e., $z \to \infty$. Indeed, they also incorporate the boundary condition far above the surface when solving the BVP, where the potential temperature field vanishes. In our case, we restrict to a limited interval of the boundary layer, i.e., $z_0 < z < z_1$, and prescribed the most generalized conditions at these boundaries, i.e., the third-kind (also known as Robin) conditions. Furthermore, while an initial condition for the potential temperature is expressed explicitly in our model, the initial condition seems to be buried inside the boundary condition for the perturbation surface heat flux in their problem formulation.

\subsection{Sturm-Liouville problem for non-steady-state perturbed potential temperature}

Let
\begin{equation*}
\theta(z,t) = \overline{\theta}(z) + \widehat{\theta}(z,t),
\end{equation*}
where $\overline{\theta}$ and $\widehat{\theta}$ denote the steady-state and non-steady-state perturbed potential temperature fields, respectively, then the former satisfies the following BVP:
\begin{align}
\frac{d}{dz} \left(u(z) \frac{d\overline{\theta}}{dz}\right) &= 0, \qquad z_0 < z < z_1, 	\label{steady}	\\
a_0 \overline{\theta}(z_0) - a_1 \frac{d\overline{\theta}}{dz}(z_0) &= c_1, 	\label{robin1} \\
b_0 \overline{\theta}(z_1) + b_1 \frac{d\overline{\theta}}{dz}(z_1) &= c_2.	\label{robin2}
\end{align} 
The ODE~\eqref{steady} solves
\begin{equation}
\overline{\theta}(z) = A \int_{z_0}^{z} \frac{d\zeta}{u(\zeta)} + B, \qquad A, B \in \mathbb{R}			\label{steady-solution}
\end{equation}
and satisfies the boundary conditions~\eqref{robin1} and~\eqref{robin2} if and only if 
\begin{align*}
a_0 B - a_1 \frac{A}{u(z_0)} &= c_1 \\
b_0 \left(A \int_{z_0}^{z_1} \frac{d\zeta}{u(\zeta)} + B \right) + b_1 \frac{A}{u(z_1)} &= c_2.
\end{align*}
Additionally, the BVP~\eqref{steady}--\eqref{robin2} has a unique solution if and only if the following condition is satisfied:
\begin{equation*}
\frac{a_1 b_0}{u(z_0)} + \frac{a_0 b_1}{u(z_1)} + a_0 b_0 \int_{z_0}^{z_1} \frac{d\zeta}{u(\zeta)} \neq 0.
\end{equation*}

On the other hand, the non-steady-state perturbed potential temperature $\widehat{\theta}$ satisfies the following IBVP:
\begin{align}
\frac{\partial \widehat{\theta}}{\partial t} = \frac{\partial}{\partial z} \left(u(z) \frac{\partial \widehat{\theta}}{\partial z} \right),& 
\qquad \qquad t > 0 \qquad z_0 < z < z_1, \label{nonsteady-perturbed-potential-ivbp} \\
\widehat{\theta}(z,0) = \widehat{\psi}(z) \equiv \psi(z) - \overline{\theta}(z),& \qquad \qquad z_0 < z < z_1, \label{nonsteady-ic} \\
a_0 \widehat{\theta}(z_0,t) - a_1 \frac{\partial \widehat{\theta}}{\partial z}(z_0,t) = 0,& \qquad \qquad t > 0, \qquad a_0^2 + a_1^2 > 0, \label{nonsteady-robin1}  \\
b_0 \widehat{\theta}(z_1,t) + b_1 \frac{\partial \widehat{\theta}}{\partial z}(z_1,t) = 0,& \qquad \qquad t > 0, \qquad b_0^2 + b_1^2 > 0.  \label{nonsteady-robin2}
\end{align}
Applying the method of separation of variables, we express $\widehat{\theta}(z,t) = y(z) \tau(t) \neq 0$. Substituting into the IVBP~\eqref{nonsteady-perturbed-potential-ivbp}--\eqref{nonsteady-robin2}, it yields the following relationship:
\begin{align*}
y(z) \tau'(t) &= \frac{d}{dz} \left(u(z) \frac{dy}{dz} \right) \tau(t) \\
\frac{\tau'(t)}{\tau(t)} &= \frac{1}{y(z)} \frac{d}{dz} \left(u(z) \frac{dy}{dz} \right) = -\lambda, \qquad \lambda \in \mathbb{R}.
\end{align*}
The price for separating the variables is one PDE in an IVBP becomes two ODEs, one in an IVP, the other in a BVP. They are given as follows, respectively:
\begin{align}
-\frac{d\tau}{dt} &= \lambda \tau, \qquad \qquad &\tau(0) &= \widehat{\psi}(z),  \\
-\frac{d}{dz} \left[u(z) \frac{dy}{dz} \right] &= \lambda y, \qquad 
&a_0 y(z_0) - a_1 \frac{dy}{dz}(z_0) &= 0, \qquad \qquad b_0 y(z_1) + b_1 \frac{dy}{dz}(z_1) = 0. \label{bvpfory}
\end{align}

Observe that the BVP~\eqref{bvpfory} is a Sturm-Liouville eigenvalue problem with infinite number of nonnegative eigenvalues $\lambda_n$, $n \in \mathbb{N}_0$, i.e., $0 \leq \lambda_0 < \lambda_1 < \cdots$. For each eigenvalue $\lambda_n$, there exists a unique eigenfunction $y_n(z)$, which forms an orthogonal set of eigenfunctions $\left\{y_n(z) \right\}$ in $(z_0,z_1)$, i.e.,
\begin{equation*}
\int_{z_0}^{z_1} y_m(z) y_n(z) \, dz = 0, \qquad \text{whenever} \; m \neq n.
\end{equation*}
For $\lambda = \lambda_n$, the ODE for $\tau$ becomes $\frac{d\tau_n}{dt} + \lambda_n \tau_n = 0$, which gives
\begin{equation*}
\tau_n(t) = C e^{-\lambda_n t}, \qquad n \geq 0, \qquad C \in \mathbb{R}.
\end{equation*} 
The solution of the non-steady state IVBP can be expressed as 
\begin{equation*}
\widehat{\theta}(z,t) = \sum_{n = 0}^{\infty} a_n y_n(z) e^{-\lambda_n t}.
\end{equation*}
Using the initial condition, we obtain
\begin{equation}
a_n = \frac{\displaystyle \int_{z_0}^{z_1} \widehat{\psi}(z) y_n(z) \, dz}{\displaystyle \int_{z_0}^{z_1} y_n^2(z) \, dz}, \qquad n \in \mathbb{N}_0. 	\label{coeff}
\end{equation}
The solution of the perturbed potential temperature reads
\begin{equation}
\theta(z,t) = \overline{\theta}(z) + \widehat{\theta}(z,t) 
= A \int_{z_0}^{z} \frac{d\zeta}{u(\zeta)} + B + \sum_{n = 0}^{\infty} a_n y_n(z) e^{-\lambda_n t},		\label{steady-nonsteady}
\end{equation}
where $A, B \in \mathbb{R}$.

We have the following corollary~\cite{agarwal2009ordinary}.
\begin{corollary}		\label{corol}
From the expression for $\theta(z,t)$~\eqref{steady-nonsteady}, we obtain the following characteristics:
\begin{enumerate}[leftmargin=1em]
\item Since each $\lambda_n > 0$, then $\theta(z,t) \to \overline{\theta}(z)$ as $t \to \infty$.
\item For any $t = \tau > 0$, the series for $\theta(z,\tau)$ converges uniformly in $z_0 \leq z \leq z_1$ because of the exponential factors. Hence, $\theta(z,\tau)$ is a continuous function in $z$.
\item For large $t$, we may approximate the perturbed potential temperature by the term containing the lowest-order eigenvalue only, i.e.,
\begin{equation*}
\overline{\theta}(z) + a_0 y_0(z) e^{-\lambda_0 t}.
\end{equation*}
\end{enumerate}
\end{corollary}

\subsection{Transformation to Liouville normal form}

Consider again the Sturm-Liouville problem~\eqref{bvpfory} in the \emph{standard} or \emph{canonical form} with eigenvalue $\lambda$ and the corresponding eigenfunction $y(z)$. We have the following lemmas.
\begin{lemma}	\label{lemma1}
The eigenvalue problem~\eqref{bvpfory} can be converted to another BVP in the following form by transforming the independent variable $z$ to $\widehat{z}$:
\begin{align*}
-\frac{d}{d\widehat{z}}\left(\frac{u}{\dot{z}} \, \frac{dy}{d\widehat{z}} \right) = \lambda \dot{z} \, y, & \qquad \qquad z_0 < z < z_1 \\
a_0 y(z_0) - \frac{a_1}{\dot{z}(\widehat{z}_0)} \frac{dy}{dz}(z_0) = 0, & \qquad \qquad b_0 y(z_1) + \frac{b_1}{\dot{z}(\widehat{z}_1)} \frac{dy}{dz}(z_1) = 0.
\end{align*}
where the dot represents the derivative with respect to $\widehat{z}$, i.e., $\dot{z} = dz/d\widehat{z}$, $y = y(z(\widehat{z}))$, and $u/\dot{z} = u(z(\widehat{z}))/\dot{z}(\widehat{z})$. 
\end{lemma}
\begin{proof}
It can be easily worked out using the fact that the differential operator $d/dz = \left(d\widehat{z}/dz \right) \, d/d\widehat{z} = (1/\dot{z}) \, d/d\widehat{z}$. Then, multiplying both sides with $\dot{z}$ yields the desired expression. The boundary conditions follow accordingly.
\end{proof}
\begin{lemma}	\label{lemma2}
The eigenvalue problem~\eqref{bvpfory} can be converted to another BVP in the following form by transforming the dependent variable $y$ of the form $y(z) = v(z) \, \widehat{y}(z)$, with $v(z)$ is a given function:
\begin{align*}
-\frac{d}{dz}\left(u v^2 \, \frac{d\widehat{y}}{dz} \right) - \frac{d}{dz}\left(u \frac{dv}{dz}\right) v \, \widehat{y} = \lambda v^2 \, \widehat{y},& \\
\alpha_0 \widehat{y}(z_0) - \alpha_1 \frac{d\widehat{y}}{dz}(z_0) = 0,& \qquad \qquad
 \beta_0 \widehat{y}(z_1) +  \beta_1 \frac{d\widehat{y}}{dz}(z_1) = 0,
\end{align*}
where
\begin{align*}
\alpha_0 &= a_0 v(z_0) - a_1 \frac{dv}{dz}(z_0), \qquad \qquad &\alpha_1 = a_1 v(z_0), \\
 \beta_0 &= b_0 v(z_0) + b_1 \frac{dv}{dz}(z_0), \qquad \qquad & \beta_1 = b_1 v(z_0).
\end{align*}
\end{lemma}
\begin{proof}
We need to apply the chain rule multiple times:
\begin{align*}
\frac{dy}{dz} &= v' \widehat{y} + v \, \widehat{y}' \\
\frac{d}{dz} \left(u \frac{dy}{dz} \right) &= \left(u v \, \widehat{y}' \right)' + u v' \widehat{y}' + \left(u v'\right)' \, \widehat{y} \\
v \frac{d}{dz} \left(u \frac{dy}{dz} \right) &= \left(u v \, \widehat{y}' \right)' v + \left(u v \widehat{y}' \right) v' + \left(u v'\right)' v \, \widehat{y}.
\end{align*}
This last expression equals to the negative of the left-hand side of the desired ODE, and the first two terms can be combined into a single term by the chain rule. The right-hand side is straightforward and the boundary conditions follow accordingly by applying the product rule.
\end{proof}

The following theorem states the transformation of the Sturm-Liouville problem from the canonical form to the Liouville normal form.
\begin{theorem}
The Sturm-Liouville problem~\eqref{bvpfory} in the canonical form with eigenvalue $\lambda$ and the corresponding eigenfunction $y(z)$ can be converted into the following BVP in the \emph{Liouville normal} or \emph{Schr\"odinger form}
\begin{align}
-\frac{d^2\widehat{y}}{d\widehat{z}^2} + q(\widehat{z}) \, \widehat{y} = \lambda \, \widehat{y},& \qquad \qquad \widehat{z}_0 < \widehat{z} < \widehat{z}_1,		\label{odeyhat}	\\
\widehat{\alpha}_0 \widehat{y}(\widehat{z_0}) - \widehat{\alpha}_1 \frac{d\widehat{y}}{d\widehat{z}}(\widehat{z_0}) = 0, & \qquad \qquad 
\widehat{ \beta}_0 \widehat{y}(\widehat{z_1}) + \widehat{ \beta}_1 \frac{d\widehat{y}}{d\widehat{z}}(\widehat{z_1}) = 0,	\label{bcyhat}
\end{align}
by performing \emph{Liouville's transformation}
\begin{equation}
\widehat{z} = \int \frac{dz}{\sqrt{u}}, \qquad v = u^{-1/4}, \qquad \text{and} \qquad y(z) = v(z) \, \widehat{y}(z),		\label{liou-trans}
\end{equation}
where
\begin{align*}
\alpha_0 &= a_0 v(\widehat{z}_0) - \frac{a_1}{\dot{z}(\widehat{z}_0)} \frac{dv}{d\widehat{z}}(\widehat{z}_0), \qquad \qquad 
&\alpha_1 = a_1 \frac{v(\widehat{z_0})}{\dot{z}(\widehat{z_0})}, \\
 \beta_0 &= b_0 v(\widehat{z}_1) + \frac{b_1}{\dot{z}(\widehat{z_1})} \frac{dv}{d\widehat{z}}(\widehat{z}_1), \qquad \qquad 
&\beta_1 = b_1 \frac{v(\widehat{z_1})}{\dot{z}(\widehat{z_1})},
\end{align*}
and $q$ is the corresponding \emph{invariant function} of the canonical PDE, given by
\begin{equation*}
q(\widehat{z}) = v \frac{d^2}{d\widehat{z}^2} \left(\frac{1}{v}\right).
\end{equation*}
\end{theorem}

\begin{proof}
Combining both ODEs in Lemmas~\ref{lemma1} and~\ref{lemma2} yields the following transformed ODE:
\begin{equation}
-\frac{d}{d\widehat{z}} \left(\frac{u}{\dot{z}} v^2 \, \frac{d\widehat{y}}{d \widehat{z}} \right) - \left(u v'\right)' v \dot{z} \, \widehat{y} =  \lambda v^2 \dot{z} \, \widehat{y}.
\end{equation}
Since $d\widehat{z}/dz = 1/\dot{z} = \sqrt{u}$ and $v^2 = u^{-1/2}$, then $uv^2/\dot{z} = 1$. Similarly, $v^2 \dot{z} = u^{-1/2} \sqrt{u} = 1$. To show that $- \left(u v'\right)' v \dot{z} = q(\widehat{z})$, we observe that
\begin{equation*}
- v \dot{z} \left(u v'\right)' 
= - v \dot{z} \frac{d}{dz} \left(\frac{\sqrt{u}}{u^{-1/2}} v '\right) 
= - v \dot{z} \frac{d}{dz} \left(\frac{\dot{z}}{v^2} v' \right)
= v \dot{z} \frac{d}{d\widehat{z}} \left(-\frac{v'}{v^2}\right) 
= v \frac{d}{d\widehat{z}} \left[\dot{z} \frac{d}{dz}\left(\frac{1}{v}\right) \right] 
= v \frac{d^2}{d\widehat{z}^2} \left(\frac{1}{v}\right).
\end{equation*}
Since the last expression on the right-hand side equals to $q(\widehat{z})$, the transformation is thus correct. We obtain the corresponding boundary conditions by combining the ones from Lemmas~\ref{lemma1} and~\ref{lemma2}. The proof is complete.
\end{proof}
\begin{corollary}
Using the Liouville transformation~\eqref{liou-trans}, the perturbed potential temperature field~\eqref{steady-nonsteady} in the transformed variables can be expressed as follows:
\begin{equation}
\theta(\widehat{z},t) = \overline{\theta}(\widehat{z}) + \sum_{n = 0}^{\infty} a_n \, v(\widehat{z}) \, \widehat{y}_n(\widehat{z}) \, e^{-\lambda_n t},		\label{theta-zhat}
\end{equation}
where the steady-state solution $\overline{\theta}$~\eqref{steady-solution} and coefficients $a_n$ in~\eqref{coeff}  are given as follows, respectively:
\begin{align*}
\overline{\theta}(\widehat{z}) &= A \int_{\widehat{z}_0}^{\widehat{z}} \frac{d\widehat{\zeta}}{\sqrt{u(\widehat{\zeta})}} + B, \qquad \qquad A, B \in \mathbb{R}, \\
a_n &= \frac{\displaystyle \int_{\widehat{z}_0}^{\widehat{z}_1} \sqrt[4]{u(\widehat{z})} \, \widehat{\psi}(\widehat{z}) \, \widehat{y}_n(\widehat{z}) \, d\widehat{z}}{\displaystyle \int_{\widehat{z}_0}^{\widehat{z}_1} \widehat{y}_n^2(\widehat{z}) \, d\widehat{z}}, \qquad \qquad n \in \mathbb{N}_0.
\end{align*}
\end{corollary}

\section{Asymptotic solution}		\label{asymptotic}

\subsection{WKB approximation}

In this subsection, we proceed in seeking the corresponding WKB solution of the BVP~\eqref{odeyhat}--\eqref{bcyhat}. In that case, we need to assume that the term $q(\widehat{z})$ is a slowly varying in such a way that its fractional change over a vertical wavelength wavelength is much less than unity. Consequently, $Q(\widehat{z}) := q(\widehat{z}) - \lambda$ is also a slowly varying function, i.e., $Q(\widehat{z})$ changes by a small fraction in a distance $1/Q$. Under this assumption, the waves locally behaves like plane waves, as if $Q$ is constant. As previously, the upper dot symbol denotes the derivative with respect to the transformed vertical variable $\widehat{z}$ and we rewrite the ODE~\eqref{odeyhat} by including a singular perturbation parameter $0 < \varepsilon \ll 1$~\cite{pryce1993numerical}:
\begin{equation}
\varepsilon^{2} \ddot{\widehat{y}} = Q(\widehat{z}) \, \widehat{y}.		\label{eps-ODE}
\end{equation}

By introducing the Riccati variable $R(\widehat{z}) = \varepsilon \dot{\widehat{y}}/\widehat{y}$, we convert the linear second-order ODE~\eqref{eps-ODE} to a nonlinear first-order (Riccati) ODE in $R$:
\begin{equation}
\varepsilon \dot{R} + R^2 = Q(\widehat{z}).		\label{eps-Riccati}
\end{equation}
By writing $R$ as an asymptotic series
\begin{equation*}
R(\widehat{z}) = R_0(\widehat{z}) + \varepsilon R_1(\widehat{z}) + \varepsilon^2 R_2(\widehat{z}) + \cdots,
\end{equation*}
and substituting it to~\eqref{eps-Riccati} and collecting the coefficients of $\varepsilon^n$, $n \in \mathbb{N}$, we obtain a set of equations for $R_n$:
\begin{align}
{\cal O}(1):& \qquad R_0^2 = Q(\widehat{z})&		\label{eqn-R0}	\\
{\cal O}(\varepsilon):& \qquad \dot{R}_0 + 2 R_0 R_1 = 0& 	\label{eqn-R1}	\\
{\cal O}(\varepsilon^2):& \qquad \dot{R}_1 + R_1^2 + 2 R_0 R_2 = 0.& 	\label{eqn-R2}
\end{align} 

From the lowest-order equation~\eqref{eqn-R0}, we obtain
\begin{equation*}
R_0(\widehat{z}) = \pm \left\{ 
\begin{array}{ll}
k(\widehat{z}), 		\qquad \text{when} \quad  Q(\widehat{z}) > 0, \\
\omega(\widehat{z}), 	\qquad \text{when} \quad  Q(\widehat{z}) < 0,
\end{array}
\right.
\end{equation*}
where we have defined $k(\widehat{z}) = \sqrt{Q(\widehat{z})}$ for $Q(\widehat{z}) > 0$ and $\omega(\widehat{z}) = \sqrt{-Q(\widehat{z})}$ for $Q(\widehat{z}) < 0$. 
Furthermore, from the Ricatti transformation, we obtain a general WKB solution of $\widehat{y}$ in terms of the Riccati variable $R$, where up to the first two terms of expansion in $R$, given as follows:
\begin{equation}
\widehat{y}(\widehat{z}) = \widehat{y}_0 \, \text{exp} \left(\frac{1}{\varepsilon} \int_{\widehat{z}_0}^{\widehat{z}} R(\zeta) \, d\zeta \right) 
= \widehat{y}_0 \, \text{exp} \left(\frac{1}{\varepsilon} \int_{\widehat{z}_0}^{\widehat{z}} R_0(\zeta) \, d\zeta \right) \, \text{exp} \left(\int_{\widehat{z}_0}^{\widehat{z}} R_1(\zeta) \, d\zeta \right). 
\end{equation}
Meanwhile, from the first-order equation in $\varepsilon$~\eqref{eqn-R1}, we obtain the relationship
\begin{equation}
\text{exp} \left(\int_{\widehat{z}_0}^{\widehat{z}} R_1(\zeta) \, d\zeta \right) = \frac{1}{\sqrt{R_0(\widehat{z})}}.
\end{equation}

One may wish to continue by including an additional term $R_2$ that can obtained from~\eqref{eqn-R2}, but for the moment, it suffices just to include only $R_0$ and $R_1$ in our WKB approximation. By putting $\varepsilon = 1$, we then obtain exponential-type and oscillatory real-valued WKB solutions, given as follows respectively:
\begin{equation}
\widehat{y}(\widehat{z}) = \left\{
\begin{array}{ll}
{\displaystyle \frac{1}{\sqrt{k(\widehat{z})}}} \left[A_1 \, \text{exp}\left({\displaystyle \int_{\widehat{z}_0}^{\widehat{z}} k(\zeta) \, d\zeta}\right) + A_2 \, \text{exp} \left({\displaystyle -\int_{\widehat{z}_0}^{\widehat{z}} k(\zeta) \, d\zeta} \right) \right], & \text{for} \quad Q(\widehat{z}) > 0, \\ [0.5cm]
{\displaystyle \frac{1}{\sqrt{\omega(\widehat{z})}}} \left[B_1 \, \text{exp} \left({\displaystyle i\int_{\widehat{z}_0}^{\widehat{z}} \omega(\zeta) \, d\zeta} \right) + B_2 \, \text{exp} \left({\displaystyle -i\int_{\widehat{z}_0}^{\widehat{z}} \omega(\zeta) \, d\zeta} \right) \right], & \text{for} \quad Q(\widehat{z}) < 0,
\end{array}
\right.		\label{expo}
\end{equation}
where $A_1, A_2, B_1, B_2 \in \mathbb{C}$ or
\begin{equation}
\widehat{y}(\widehat{z}) = \left\{ 
\begin{array}{ll}
{\displaystyle \frac{1}{\sqrt{k(\widehat{z})}}} \left[\widehat{A}_1 \cosh \left({\displaystyle \int_{\widehat{z}_0}^{\widehat{z}} k(\zeta) \, d\zeta}\right) + \widehat{A}_2 \sinh \left({\displaystyle \int_{\widehat{z}_0}^{\widehat{z}} k(\zeta) \, d\zeta}\right) \right], &\text{for} \quad Q(\widehat{z}) > 0, \\ [0.5cm]
{\displaystyle \frac{1}{\sqrt{\omega(\widehat{z})}}} \left[\widehat{B}_1 \cos \left({\displaystyle \int_{\widehat{z}_0}^{\widehat{z}} \omega(\zeta) \, d\zeta} \right) + \widehat{B}_2 \sin \left({\displaystyle \int_{\widehat{z}_0}^{\widehat{z}} \omega(\zeta) \, d\zeta} \right) \right], &\text{for} \quad Q(\widehat{z}) < 0,
\end{array}
\right.		\label{trigo}
\end{equation}
where $\widehat{A}_1, \widehat{A}_2, \widehat{B}_1, \widehat{B}_2 \in \mathbb{R}$.

\subsection{Turning point}		\label{turningpoint}

The above WKB approximations~\eqref{expo} or~\eqref{trigo} are valid as long as $Q(\widehat{z}) \neq 0$. A point $\widehat{z} = \zeta_c$ where $q(\zeta_c) = \lambda$ but $\dot{q}(\zeta_c) \neq 0$ is called a (simple) \emph{turning point} of the ODE~\eqref{odeyhat}. Similar as previously, we parameterize the ODE as in~\eqref{eps-ODE}, where $\varepsilon$ is again a small parameter and $\lambda$ is a fixed eigenvalue:
\begin{equation}
-\varepsilon^2 \, \ddot{\widehat{y}} +  \left[q(\widehat{z}) - \lambda \right] \widehat{y} = 0.	\label{eps-tp}
\end{equation}
We consider a particular case of a single turning point at $\widehat{z} = \zeta_c$ with $\dot{q}(\zeta_c) > 0$. In this case, the WKB solutions are oscillating to the left of $\widehat{z} = \zeta_c$ while take exponentially decaying behavior to the right of $\widehat{z} = \zeta_c$. Now, introduce the change of variable $\delta \, t = \widehat{z} - \zeta_c$, where $\delta^3 = \varepsilon^2/\dot{q}(\zeta_c)$. The ODE~\eqref{eps-tp} becomes
\begin{equation}
-\frac{d^2\widehat{y}}{dt^2} + \left(\frac{q(\zeta_c + \delta t) - \lambda}{\delta \, \dot{q}(\zeta_c)}\right) \widehat{y} = 0.	\label{delta-tp}
\end{equation}
As $\varepsilon \to 0$, $\delta \to 0$, and by Taylor-expanding $q$ about the turning point, i.e., $q(\zeta_c + \delta t) = q(\zeta_c) + \dot{q}(\zeta_c) \delta t + {\cal O}(\delta^2)$, the ODE~\eqref{delta-tp} approximates to the Airy (or Stokes) ODE:
\begin{equation}
\frac{d^2\widehat{y}}{dt^2} = t \widehat{y}.		\label{airy}
\end{equation}

The solution of Airy equation~\eqref{airy} is given by a linear combinations of Airy functions:
\begin{equation}
\widehat{y}(t) = C_1 \, \text{Ai}(t) + C_2 \, \text{Bi}(t),	\label{airy-sol}
\end{equation}
where $C_1, C_2 \in \mathbb{C}$, and Ai$(t)$ and Bi$(t)$ are the Airy function of the first and second kinds, respectively. They can be defined by the improper Riemann integral and their convergence can be verified using Dirichlet's test:
\begin{align*}
\text{Ai}(t) &= \frac{1}{\pi} \int_{0}^{\infty} \cos \left(\frac{x^3}{3} + t x \right) \, dx, \\
\text{Bi}(t) &= \frac{1}{\pi} \int_{0}^{\infty} \left[\text{exp} \left(-\frac{x^3}{3} + t x \right) + \sin \left(\frac{x^3}{3} + t x \right) \right] \, dx.
\end{align*}
These functions are also called \emph{connection} or \emph{turning point functions} since they change their characteristics from exponential decay when $t \to +\infty$ to oscillation when $t \to -\infty$. As we will observe very soon, they also connect or ``bridge'' between the two different natures of the WKB solutions~\eqref{expo}.

The asymptotic expansions of Ai$(t)$ and Bi$(t)$ for $t \to \pm \infty$ can be computed either by extending them to the complex plane or directly using the method of stationary phase from their corresponding Riemann integral representations~\cite{schleich2001quantum}. We provide the results here~\cite{abramowitz1964handbook}. As $t \to \infty$
\begin{align}
\text{Ai}(t) \sim \frac{1}{2 \sqrt{\pi} t^{1/4}} e^{-\frac{2}{3} t^{3/2}} \qquad \qquad \text{and} \qquad \qquad
\text{Bi}(t) \sim \frac{1}{  \sqrt{\pi} t^{1/4}} e^{ \frac{2}{3} t^{3/2}},
\end{align}
and as $t \to -\infty$
\begin{align}
\text{Ai}(t) \sim \frac{1}{\sqrt{\pi} |t|^{1/4}} \sin \left(\frac{2}{3} |t|^{3/2} + \frac{\pi}{4} \right) \qquad \qquad \text{and} \qquad \qquad
\text{Bi}(t) \sim \frac{1}{\sqrt{\pi} |t|^{1/4}} \cos \left(\frac{2}{3} |t|^{3/2} + \frac{\pi}{4} \right).
\end{align}

We write the WKB solutions~\eqref{expo} and the solution of the Airy ODE~\eqref{airy-sol} in relation to the regions where a turning point is located, where the subscripts $L$, $C$, and $R$ denote left, center (transition), and right of the turning point $\zeta_c$, respectively:
\begin{align*}
\widehat{y}_L(\widehat{z}) &= \frac{1}{\sqrt{\omega(\widehat{z})}} \left[B_1 \, \text{exp} \left({\displaystyle i\int_{\widehat{z}_0}^{\widehat{z}} \omega(\zeta) \, d\zeta} \right) + B_2 \, \text{exp} \left({\displaystyle -i\int_{\widehat{z}_0}^{\widehat{z}} \omega(\zeta) \, d\zeta} \right) \right], \\
\widehat{y}_C(\widehat{z}) &= C_1 \, \text{Ai}\left(\frac{\widehat{z} - \zeta_c}{\delta} \right) + C_2 \, \text{Bi}\left(\frac{\widehat{z} - \zeta_c}{\delta} \right),  \\
\widehat{y}_R(\widehat{z}) &= \frac{1}{\sqrt{k(\widehat{z})}} \left[A_1 \, \text{exp}\left({\displaystyle \int_{\widehat{z}_0}^{\widehat{z}} k(\zeta) \, d\zeta}\right) + A_2 \, \text{exp} \left({\displaystyle -\int_{\widehat{z}_0}^{\widehat{z}} k(\zeta) \, d\zeta} \right) \right].
\end{align*}
If we choose the constants appropriately, then the above solutions match smoothly from one region to another. By taking the outer limits of the inner (Airy) solution and the inner limits of the outer (WKB) solutions, we obtain
\begin{align*}
\widehat{y}_L(\widehat{z}) &\sim \left[\dot{q}(\zeta_c) \left(\zeta_c - \widehat{z}\right)\right]^{-1/4} \left(B_1 e^{-i\theta} + B_2 e^{i\theta} \right), & \text{as} \; \widehat{z} \to \zeta_c^{-}, \\
\widehat{y}_C(\widehat{z}) &\sim \frac{\sqrt[6]{\varepsilon}}{\sqrt{\pi}} \left[\sqrt[3]{\dot{q}(c)} \left(\zeta_c - \widehat{z}\right)\right]^{-1/4} 
\left[C_1 \sin\left(\theta + \frac{\pi}{4} \right) + C_2 \cos \left(\theta + \frac{\pi}{4} \right)\right],
& \qquad \text{as} \; \widehat{z} - \zeta_c \to -\infty, \\
\widehat{y}_C(\widehat{z}) &\sim \frac{\sqrt[6]{\varepsilon}}{\sqrt{\pi}} \left[\sqrt[3]{\dot{q}(\zeta_c)} \left(\widehat{z} - \zeta_c \right)\right]^{-1/4} \left(\frac{C_1}{2} e^{-\theta} + C_2 e^{\theta} \right), & \text{as} \; \widehat{z} - \zeta_c \to +\infty, \\
\widehat{y}_R(\widehat{z}) &\sim \left[\dot{q}(\zeta_c) \left(\widehat{z} - \zeta_c\right)\right]^{-1/4} \left(A_1 e^{ \theta} + A_2 e^{-\theta} \right), & \text{as} \; \widehat{z} \to \zeta_c^{+},
\end{align*}
where $\theta = \frac{2}{3\varepsilon} \sqrt{\dot{q}(\zeta_c) |\widehat{z} - \zeta_c|^3}$. These expressions match smoothly as long as the following conditions are satisfied (see also~\cite{ablowitz2003complex}):
\begin{align*}
A_1 &= \frac{C_2}{\sqrt{\pi}} \sqrt[6]{\varepsilon \, \dot{q}(\zeta_c)}, \qquad \qquad
&B_1 = \frac{\sqrt[6]{\varepsilon \, \dot{q}(\zeta_c)}}{2 \sqrt{\pi}} \left(C_1 e^{i \frac{\pi}{4}} + C_2 e^{-i \frac{\pi}{4}} \right), \\
A_2 &= \frac{C_1}{2 \sqrt{\pi}} \sqrt[6]{\varepsilon \, \dot{q}(\zeta_c)}, \qquad \qquad
&B_2 = \frac{\sqrt[6]{\varepsilon \, \dot{q}(\zeta_c)}}{2 \sqrt{\pi}} \left(C_1 e^{-i \frac{\pi}{4}} + C_2 e^{i \frac{\pi}{4}} \right).
\end{align*}

\subsection{Eigenvalue estimate}

There are several ways in estimating the eigenvalues. The first approach is based on the geometrical optics approximation of the WKB technique and adopts an assumption that $\lambda - q(\widehat{z}) > 0$ on $[\widehat{z}_0,\widehat{z}_1]$. The $k$th eigenvalue $\lambda_k$ is the root of the equation~\cite{bailey2001sleighn2}:
\begin{equation}
\int_{\widehat{z}_0}^{\widehat{z}_1} \sqrt{\lambda_k - q(\zeta)} \, d\zeta = (k + 1) \pi, \qquad k \in \mathbb{N}.
\end{equation}
The second way is appropriate for problems with two turning points, known as the classical approximation quantization formula. The $k$th eigenvalue $\lambda_k$ is approximated by the root of
\begin{equation}
\int_{\zeta_{c_1}}^{\zeta{c_2}} \sqrt{\lambda_k - q(\zeta)} \, d\zeta = \left(k + \frac{1}{2} \right) \pi, \qquad k \in \mathbb{N},
\end{equation}
where $\zeta_{c_1} < \zeta_{c_2}$ are the two turning points inside the interval for which $\lambda - q(\widehat{z}) \geq 0$ and to be found as part of the process~\cite{pryce1993numerical}. However, it is unclear how to obtain three unknowns with this single equation, and the author did not elaborate further on this matter.

A new technique for approximating the eigenvalues and the subregions that support the corresponding localized eigenfunctions was proposed by Arnorld et al.~\cite{arnold2019computing}. The approach is based on theoretical propositions of the localization landscape and effective potential proposed by Filoche and Mayboroda~\cite{filoche2016universal}. The authors introduced a \emph{landscape function} $\widehat{v}$ that satisfies the following ODE
\begin{equation}
-\frac{d^2\widehat{v}}{d\widehat{z}^2} + q(\widehat{z}) \widehat{v} = 1  
\end{equation}
with appropriate boundary conditions. Furthermore, an \emph{effective potential} $\widehat{V}$ is defined as the reciprocal of the landscape function, i.e., $\widehat{V} = 1/\widehat{v}$. The lowest eigenvalue of the Sturm-Liouville problem is given by
\begin{equation}
\lambda_0 \approx \frac{5}{4} \widehat{V}_\text{min}
\end{equation}
where $\widehat{V}_\text{min}$ is the minimum value of the effective potential on a subdomain where the corresponding eigenfunction is localized.

\section{Numerical simulation}		\label{numerical}

In this section, we consider a regular Sturm-Liouville problem with Dirichlet boundary conditions and solve the Liouville normal form numerically. We implemented a fourth-order method of boundary value problem solver \texttt{bvp4c} from \emph{Matlab} version R2021b. This solver implements the three-stage Lobatto-IIIA formula in the finite difference technique. The formula implement the collocation technique that utilizes a mesh of points in dividing the interval of integration into subintervals. The corresponding collocation polynomial then provides continuous not only solutions but also its first derivative with uniform fourth-order accuracy within the interval of integration~\cite{kierzenka2001abvp,shampine2003solving,jay2015lobatto}.
\begin{figure}[h]
\begin{center}
\includegraphics*[width=0.45\textwidth]{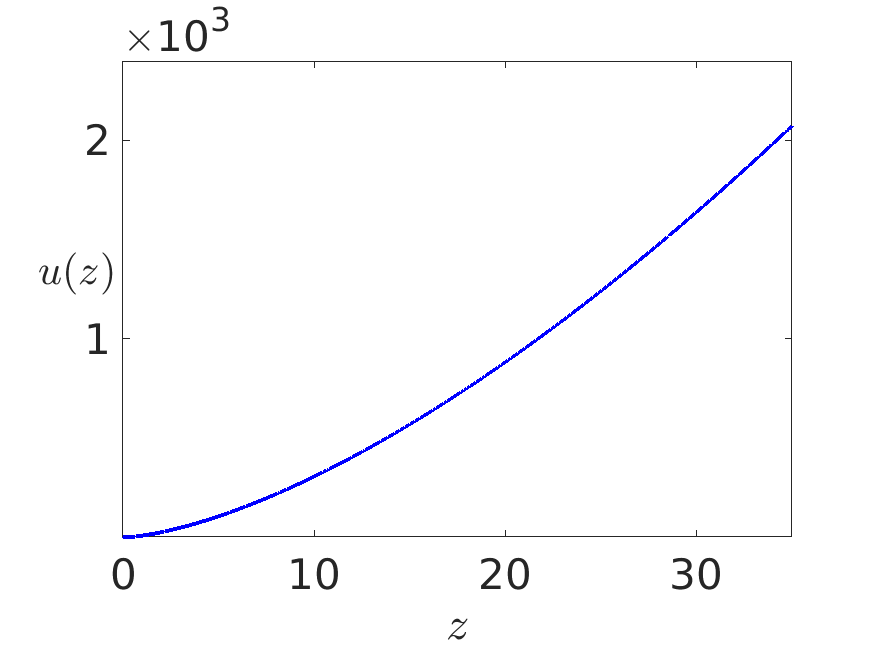} \hspace*{0.5cm}
\includegraphics*[width=0.45\textwidth]{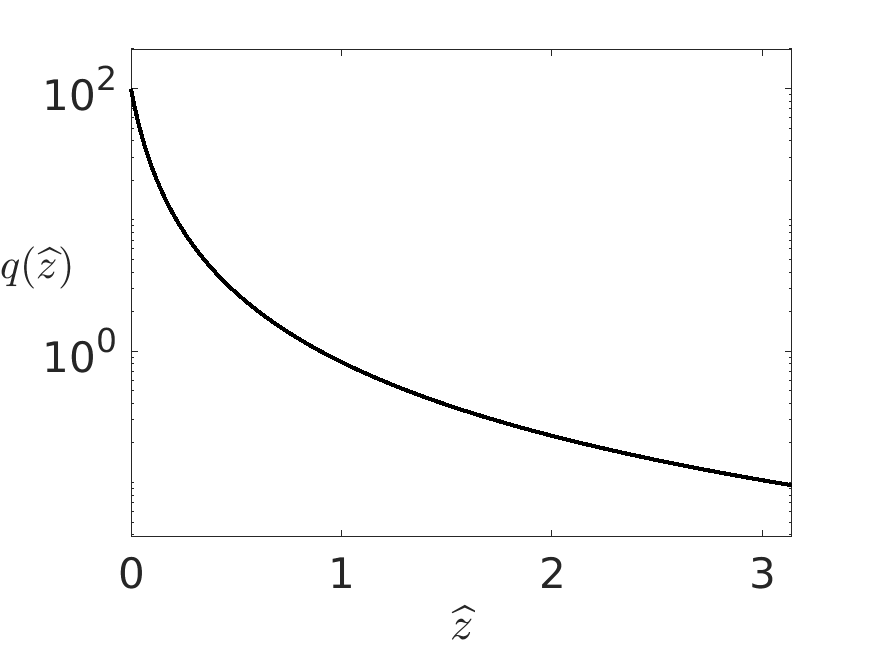}
\end{center}
\caption{Plots of the (a) wind speed profile $u(z)$ and (b) invariant function $q(\widehat{z})$ in logarithmic scale.}	\label{fig1-paine}
\end{figure}

Consider the following regular Sturm-Liouville problem with an exponential-type horizontal mean wind speed profile $u(z)$ and Dirichlet boundary conditions: 
\begin{equation*}
-\frac{d}{dz} \left\{ \left[\left(2 + \sqrt{5} \right)(z - d)\right]^{2(3 - \sqrt{5})} \frac{dy}{dz} \right\} = \lambda y, \qquad \qquad y(0) = 0 = y(z_1). 
\end{equation*}
where $d = -(0.1)^{2 + \sqrt{5}}/(2 + \sqrt{5}) \approx 1.37 \times 10^{-5}$ and $z_1 = (\pi + 0.1)^{2 + \sqrt{5}}/(2 + \sqrt{5}) + d \approx 34.4068$. By implementing the Liouville transformation, the corresponding BVP in the Liouville normal form reduces to the well-known second Paine-de Hoog-Anderssen (PdHA) problem, given as follows~\cite{paine1981correction}:
\begin{equation*}
-\frac{d^2\widehat{y}}{d\widehat{z}^2} + \frac{1}{(\widehat{z} + \frac{1}{10})^2} \widehat{y} = \lambda \widehat{y}, \qquad \widehat{y}(0) = 0 = \widehat{y}(\pi).
\end{equation*}
Figure~\ref{fig1-paine} displays the horizontal mean wind speed profile $u(z)$ and the corresponding invariant function $q(\widehat{z})$. The latter is plotted in logarithmic scale.

It turns out that this problem exhibits a single turning point at $\zeta_c = 0.7111$ but $\dot{q}(\zeta_c) = -3.7481 < 0$ instead of positive as we discussed earlier in Subsection~\ref{turningpoint}. Thus, the corresponding WKB solution will have the opposite characteristics from the one we in Section~\ref{asymptotic}, i.e., exponentially decaying to the left of $\zeta_c$ and oscillating to the right of $\zeta_c$. Since there are several families of solution, we make further restriction by imposing the condition $\dot{\widehat{y}}(0) = 1$ to obtain a unique solution. As previously, the subscript $L$, $C$, and $R$ denote the region on the left-hand side, the center (or around), and the right-hand side of the turning point $\zeta_c$. It is given explicitly as follows:
\begin{align*}
\widehat{y}_L(\widehat{z}) &= \frac{1}{\sqrt{k(0) k(\widehat{z})}} \sinh \left(\int_{0}^{\widehat{z}} k(\zeta) \, d\zeta \right), \\
\widehat{y}_C(\widehat{z}) &= C_1 \, \text{Ai} \left(\frac{\zeta_c - \widehat{z}}{\delta} \right) + C_2 \, \text{Bi} \left(\frac{\zeta_c - \widehat{z}}{\delta} \right),\\
\widehat{y}_R(\widehat{z}) &= \frac{2B_1}{\sqrt{\omega(\widehat{z})}} e^{i \left(\frac{\pi}{2} + \phi \right)} \sin \left(\int_{0}^{\widehat{z}} \omega(\zeta) \, d\zeta + \phi \right),
\end{align*}
where $B_1 \in \mathbb{C}$, the phase $\phi$, and the coefficients $C_1$ and $C_2$ are given by
\begin{align*}
\phi &= \int_{0}^{\pi} \omega(\zeta) \, d\zeta, \qquad \qquad 
C_1   = \frac{2B_1 \sqrt{\pi}}{\sqrt[6]{\varepsilon |\dot{q}(\zeta_c)| }}, \qquad \qquad
C_2   = -\frac{1}{2} \sqrt{ \frac{\pi}{k(0)} } \left(\varepsilon |\dot{q}(\zeta_c)| \right)^{-1/6}.
\end{align*}

Figure~\ref{fig2-wkb-numeric} displays the plots of eigenfunction for the second PdHA problem corresponding to the lowest eigenvalue. The left panel is obtained from the WKB approximation while the right panel is obtained numerically with an additional boundary condition $\dot{\widehat{y}}(0) = 1$. 
\begin{figure}[h]
\begin{center}
\includegraphics*[width=0.45\textwidth]{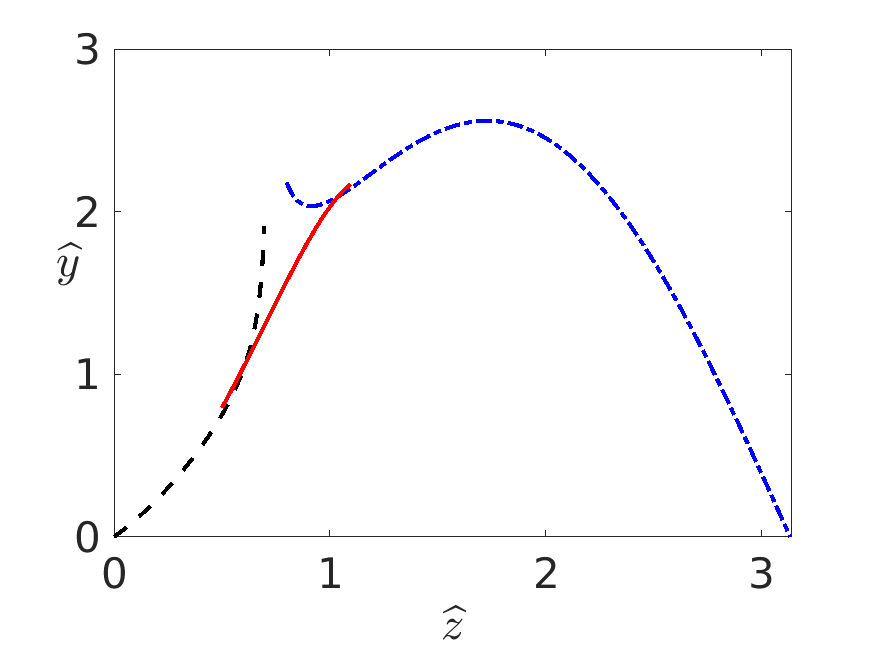} \hspace*{0.5cm}
\includegraphics*[width=0.45\textwidth]{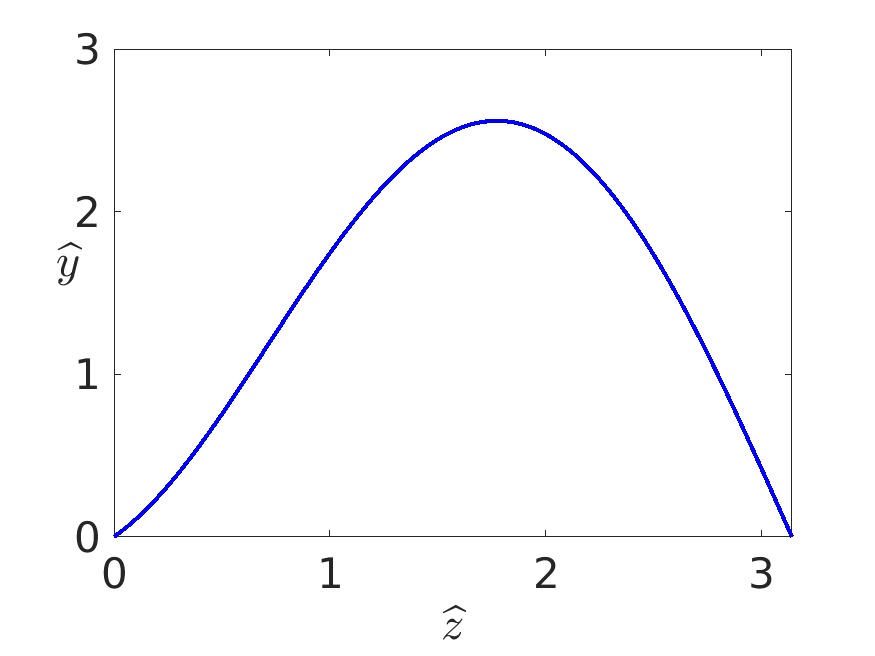}
\end{center}
\caption{(a) The WKB approximation of the eigenfunction for the lowest eigenvalue for the second PdHA BVP. The exponential profile (dashed black) and the sinusoidal solution (dash dotted blue) are connected by the Airy function (solid red) where there is a turning point at $\zeta_c = 0.7111$. For this particular plot, we have taken $\delta = 0.45$ and $|B_1| = 2.6944$. (b) An obtained numerically plot of the lowest eigenfunction $\widehat{y}_0(\widehat{z})$ with the associated eigenvalue $\lambda_0 \approx 1.520$ and maximum value of $2.558$.}	\label{fig2-wkb-numeric}
\end{figure}

It is interesting to note that this second PdHA problem is one of the few BVPs where estimating the lowest eigenvalue $\lambda_0$ can be performed almost straightforwardly. This is due to the availability of an analytical solution for the landscape function $\widehat{v}$, where the corresponding BVP is given as follows:
\begin{equation*}
-\frac{d^2\widehat{v}}{d\widehat{z}^2} + \frac{\widehat{v}}{(\widehat{z} + \frac{1}{10})^2} = 1, \qquad \widehat{v}(0) = 0 = \widehat{v}(\pi).
\end{equation*}
It admits an exact solution in the following form:
\begin{equation*}
v(\widehat{z}) = C_1 (10 \widehat{z} + 1)^\varphi + C_2 (10 \widehat{z} + 1)^{1 - \varphi} - \widehat{z}^2 - \frac{\widehat{z}}{5} - \frac{1}{100}.
\end{equation*}
where $\varphi$ and $1 - \varphi$ are the golden ratio and the negative of silver ratio, respectively, and
\begin{align*}
C_1 &= \frac{1}{100} \frac{(10 \pi + 1)^2 - (10 \pi + 1)^{1 - \varphi}}{(10 \pi + 1)^\varphi - (10 \pi + 1)^{1 - \varphi}}, \qquad \qquad
&\varphi &= \frac{1}{2}(1 + \sqrt{5}), \\
C_2 &= \frac{1}{100} \frac{(10 \pi + 1)^\varphi - (10 \pi + 1)^2}{(10 \pi + 1)^\varphi - (10 \pi + 1)^{1 - \varphi}}, \qquad \qquad 
&1 - \varphi &= \frac{1}{2}(1 - \sqrt{5}) = -\frac{1}{\varphi}.
\end{align*}
In this example, we have $v_\text{max} \approx 0.8145$, $V_\text{min} \approx 1.2277$, and thus $\lambda_0 \approx 1.5347$. Using the geometrical optics approximation of the WKB method, the approximated eigenvalue gives $\lambda_0 \approx 1.5282$ while the numerical result gives $\lambda_0 \approx 1.5198$.

The left panel of Figure~\ref{fig3-landscape} displays plots of the landscape function $\widehat{v}(\widehat{z})$, its effective potential $\widehat{V}(\widehat{z})$, and the approximated lowest-order eigenvalue $\lambda_0$. The right panel of Figure~\ref{fig3-landscape} is similar to the right panel of Figure~\ref{fig1-paine}, but not it is plotted in a regular scale. The lowest-order eigenvalue obtained from numerical simulation is also depicted for a comparison.
\begin{figure}[h]
\begin{center}
\includegraphics*[width=0.45\textwidth]{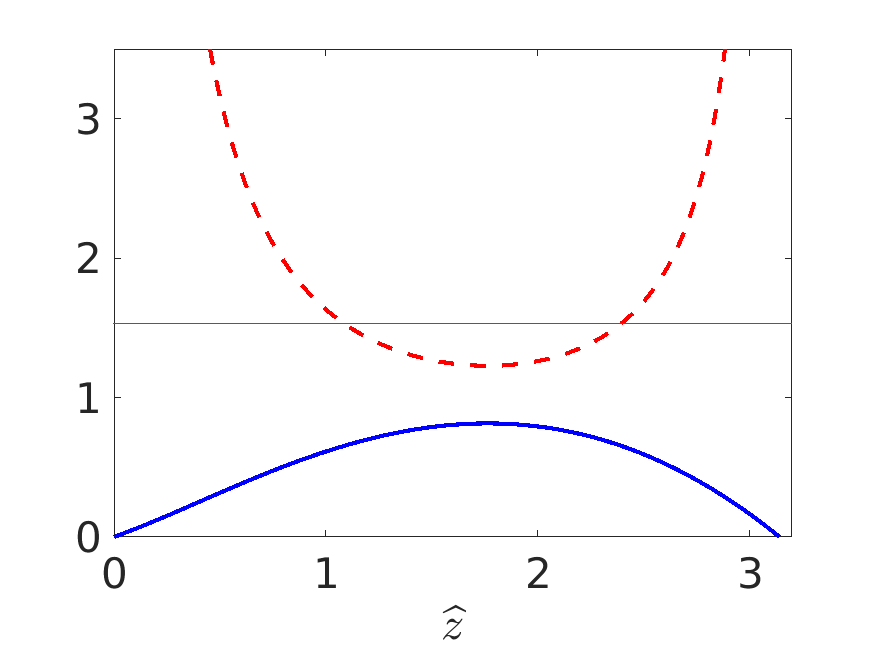} \hspace*{0.5cm}
\includegraphics*[width=0.45\textwidth]{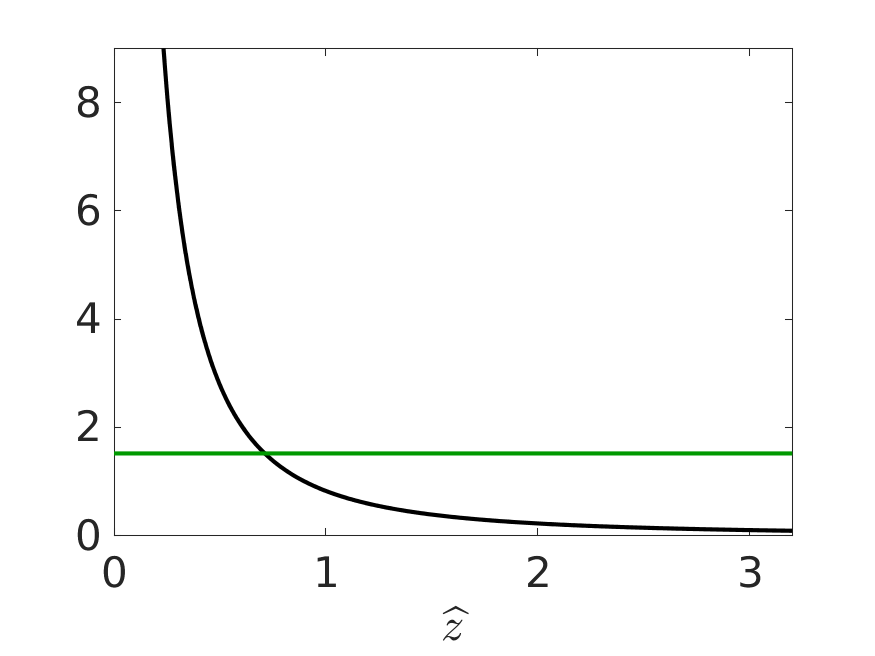}
\end{center}	
\caption{(a) The plot of a landscape function $\widehat{v}(\widehat{z})$ (solid blue), effective potential $\widehat{V}(\widehat{z})$ (dashed red), and the approximated eigenvalue $\lambda_0 \approx 1.5347$ (thin green). (b) The plot of $q(\widehat{z}) = 1/\left(\widehat{z} + 0.1\right)^2$ (black) and the numerically obtained eigenvalue $\lambda_0 = 1.5198$ (green).}		\label{fig3-landscape}
\end{figure}

By considering the original IVBP for the perturbed potential temperature field for atmospheric boundary layer, i.e., the model~\eqref{pde-theta}--
\eqref{bc2-theta},  we can acquire both the steady-state and non-steady-state solutions of $\theta$. Since $u(z) = \left[\left(2 + \sqrt{5} \right)(z - d)\right]^{2(3 - \sqrt{5})}$, it follows that $u(\widehat{z}) = (\widehat{z} + 0.1)^{2(1 + \sqrt{5})}$. The steady-state solution $\overline{\theta}(\widehat{z})$ can be found instantly
\begin{equation*}
\overline{\theta}(\widehat{z}) = A \int_{0}^{\widehat{z}} (\widehat{\zeta} + 0.1)^{-(1 + \sqrt{5})} \, d\widehat{\zeta} + B 
= \frac{c_2 - c_1}{1 - (10 \pi + 1)^{-\sqrt{5}}} \left(1 - \frac{1}{(10 \widehat{z} + 1)^{\sqrt{5}}} \right) + c_1,
\end{equation*}

Figure~\ref{fig4-3d} displays the three-dimensional plots of the perturbed potential temperature $\theta(\widehat{z},t)$, i.e., expression~\eqref{theta-zhat}, for relatively ``large'' value of $t$, where we consider the lowest-value of eigenvalue $\lambda_0$ only. See the third point of Corollary~\ref{corol}. The left and right panels illustrate the field in the absence and presence of the steady-state solution, i.e., $\widehat{\theta}$ and $\overline{\theta} + \widehat{\theta}$, respectively, albeit with a different angle point of view. In both instances, the perturbed potential temperature profiles exhibit exponentially decaying behavior as $t \to \infty$ due to the positive value of the lowest eigenvalue $\lambda_0$.
\begin{figure}[h]
\begin{center}
\includegraphics[width=0.45\textwidth]{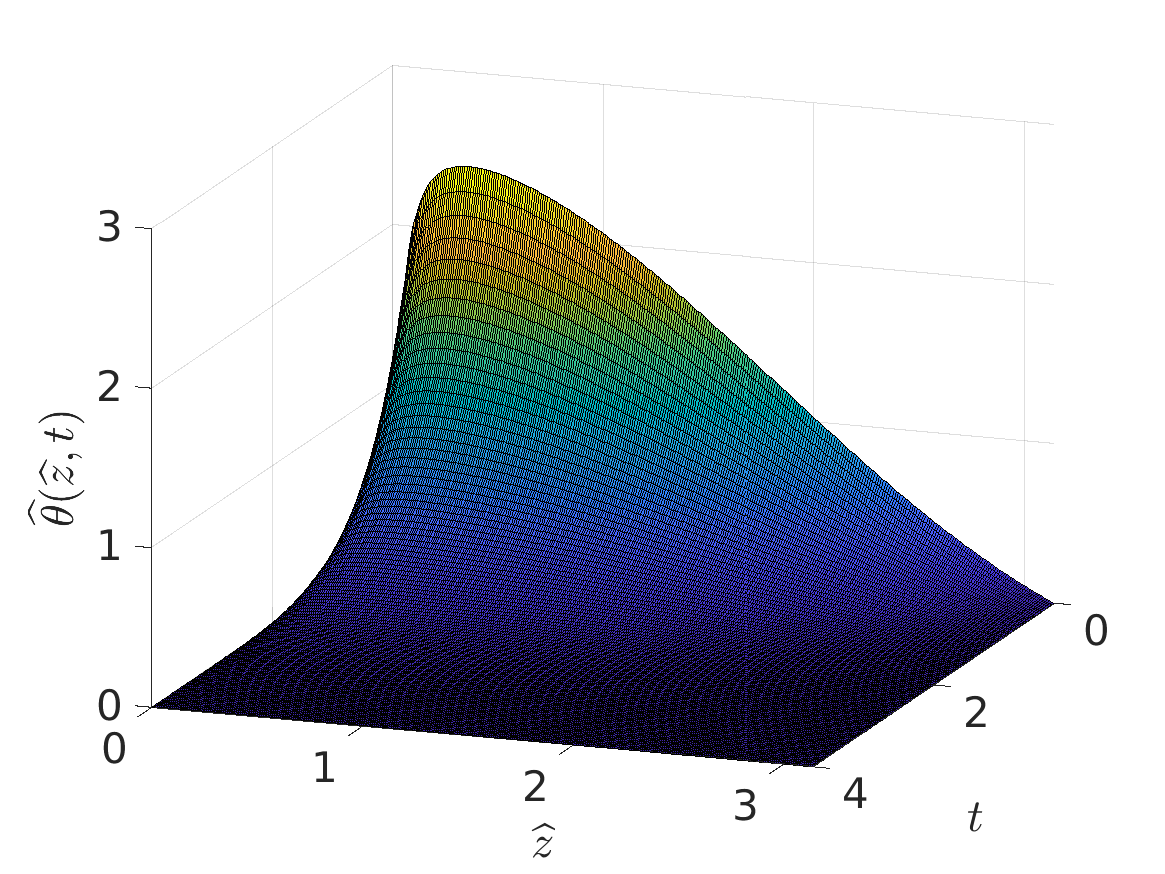}	\hspace*{0.5cm}
\includegraphics[width=0.45\textwidth]{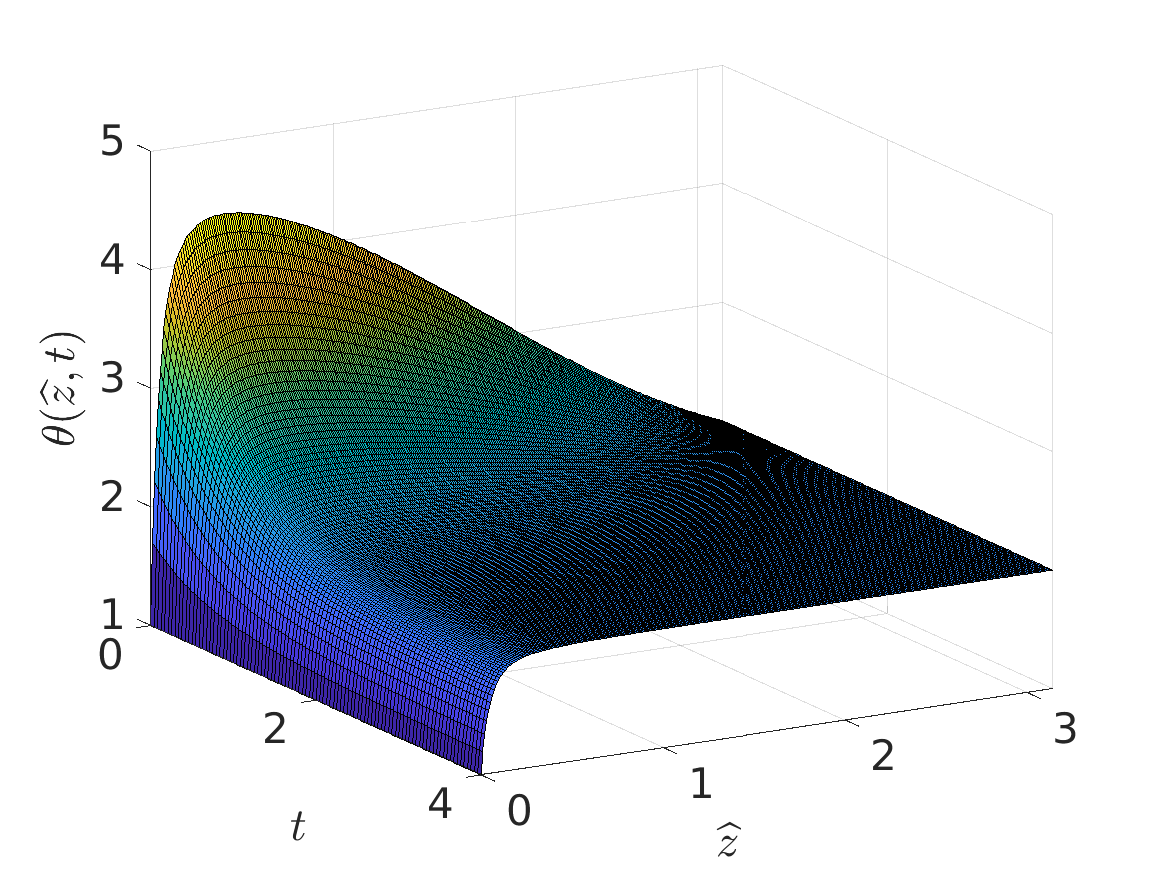}
\end{center}
\caption{Three-dimensional plots of the perturbed potential temperature $\theta(\widehat{z},t)$, (a) in the absence of the steady-state solution and (b) in the presence of the steady-state solution. In both plots, we have taken $v(\widehat{z}) = (\widehat{z} + 0.1)^{-\varphi}$, $\psi(\widehat{z}) = 1 + \widehat{z}(\widehat{z} - 2h)/[\pi(\pi - 2h)]$, $h = 1.8$, $c_1 = 1$, $c_2 = 2$, and $a_0 = 1.3808$.}			\label{fig4-3d}
\end{figure}

\section{Conclusion}		\label{conclusion}

In this article, we have considered a mathematical model for the perturbed potential temperature field in atmospheric boundary layers. The governing equation is formulated as an initial boundary value problem, and the spatial-dependent quantity obeys a regular Sturm-Liouville eigenvalue problem. We have discussed asymptotic solutions of the corresponding boundary value problem in the Liouville normal form through the WKB approximation. Furthermore, by implementing a fourth-order method solver, we also solve the problem numerically. By comparing the approximated eigenvalues obtained through the landscape function and effective potential with the ones from numerical simulations, we observed that both outcomes are reliable and remarkably in full agreement. The eigenfunctions obtained from the WKB method and numerical simulation also exhibit good qualitative behavior even though the interval of interest contains a simple turning point.

\end{document}